\documentclass[10pt, conference]{IEEEtran}
\usepackage[utf8]{inputenc}
\usepackage[russian,english]{babel}

\usepackage{graphicx}

\usepackage{enumerate}
\usepackage{url}
\usepackage{latexsym,amssymb,amsfonts,amsmath, amsthm, amscd, euscript, MnSymbol}
\usepackage{mathtools}
\usepackage{ifthen}
\newtheorem{theorem}{Theorem}
\newtheorem{definition}[theorem]{Definition}
\newtheorem{definitions}[theorem]{Definitions}

\newtheorem{remark}[theorem]{Remark}

\newtheorem{corollary}[theorem]{Corollary}
\newtheorem{proposition}[theorem]{Proposition}
\newtheorem{lemma}[theorem]{Lemma}

\def\2{{\bf 2}}
\def\P2{{\rm Par}(\2)}

\def\O2{{\rm Op}(\2)}
\def\Pn2{{\rm Par}^{(n)}(2)}

\def\N{\mathds{N}}

\def\vv{\lower 4pt \hbox{$\buildrel
{\textstyle{v}}\over{\scriptstyle{\sim}}$}}
\def\uu{\lower 4pt \hbox{$\buildrel
{\textstyle{u}}\over{\scriptstyle{\sim}}$}}
\def\ww{\lower 4pt \hbox{$\buildrel
{\textstyle{w}}\over{\scriptstyle{\sim}}$}}
\def\r{\rho}

\def\3{{\bf 3}}

\def\J2{{\rm J}(\2)}
\def\N{{\Bbb N}}

\def\dom{{\rm dom \;}}
\def\powerset{{\mathcal P}}

\def\uu{{\bf u}}
\def\vv{{\bf v}}
\def\ww{{\bf w}}

\def\0{{\bf 0}}
\def\1{{\bf 1}}

\DeclareMathOperator{\Lin}{Lin}
\DeclareMathOperator{\Tour}{Tour}
\DeclareMathOperator{\im}{im}
\DeclareMathOperator{\ded }{ded}

\begin{document}

\title{Hereditary rigidity, separation and density\\ {\large  In memory of  Professor I.G.
Rosenberg.}}

\author{
% \IEEEauthorblockN{\hspace{3cm}}
\and
\IEEEauthorblockN{Lucien  Haddad \hspace{2.5cm}}%
%%\thanks{Lucien Haddad is supported by RMC Academic Research Programme}
\IEEEauthorblockA
{Dept. of Mathematics \& CS\hspace{2.5cm}\\
  Royal Military College of Canada \hspace{2.5cm}\\
Kingston, Ontario, Canada \hspace{2.5cm}\\  Email: haddad-l@rmc.ca \hspace{2.5cm}}
\and
\IEEEauthorblockN{\hspace{2cm}}
\and
\IEEEauthorblockN{Masahiro Miyakawa \hspace{-1cm}} %
\IEEEauthorblockA{Tsukuba University of Technology\hspace{-1cm}\\
4-12-7 Kasuga, Tsukuba,
\hspace{-1cm}\\
Ibaraki 305-8521, Japan \hspace{-1cm}\\
Email: mamiyaka@cs.k.tsukuba-tech.ac.jp\hspace{-1cm}}

 \and
 \IEEEauthorblockN{\hspace{2cm}}
                         \and
                  \IEEEauthorblockN{Maurice Pouzet \hspace{1.8cm}}%
\IEEEauthorblockA{ICJ, Univ. Claude-Bernard Lyon1,  
 Villeurbanne, France\\ Dept. of Math \& Stat, 
 Univ. of Calgary,\\ Calgary, Alberta, Canada  \hspace{1cm}\\
  Email: pouzet@univ-lyon1.fr \hspace{1cm}}
 \and
\IEEEauthorblockN{\hspace{2cm}}
\and
      \IEEEauthorblockN{Hisayuki Tatsumi \hspace{1cm}}%
\IEEEauthorblockA{Tsukuba University of Technology\hspace{1cm}\\
4-12-7 Kasuga, Tsukuba,
\hspace{1cm}\\
Ibaraki 305-8521, Japan \hspace{1cm}\\
Email: tatsumi@cs.k.tsukuba-tech.ac.jp\hspace{1cm}}}
\maketitle
      
\begin{abstract} We continue the investigation of systems of hereditarily rigid relations started in Couceiro, Haddad, Pouzet and Sch\"olzel \cite{couceiro-haddad-pouzet-scholzel}. We observe that on a set $V$ with $m$ elements, there is a hereditarily rigid set $\mathcal R$ made of $n$ tournaments if and only if $m(m-1)\leq 2^n$. We ask if the same inequality  holds  when the tournaments are replaced by linear orders. This problem has  an  equivalent formulation in terms  of separation of linear orders. Let $h_{\Lin}(m)$ be the least cardinal $n$ such that there is a family $\mathcal R$ of $n$ linear orders on an $m$-element set $V$ such that any two distinct ordered pairs of distinct elements of $V$ are separated by some member of $\mathcal R$,  then  $ \lceil \log_2 (m(m-1))\rceil\leq h_{\Lin}(m)$  with equality if $m\leq 7$. We ask whether the equality holds for every $m$.  We prove that $h_{\Lin}(m+1)\leq h_{\Lin}(m)+1$. If $V$ is infinite, we show that $h_{\Lin}(m)= \aleph_0$ for $m\leq 2^{\aleph_0}$.  More generally, we prove that the two  equalities   $h_{\Lin}(m)= log_2 (m)= d(\Lin(V))$ hold,  where $\log_2 (m)$ is the least cardinal $\mu$ such that $m\leq 2^\mu$, and  $d(\Lin(V))$ is the topological density of the set  $\Lin(V)$ of linear orders on $V$ (viewed as a subset of the  power set $\powerset (V\times V)$ equipped with the product topology). These equalities follow from the {\it Generalized Continuum Hypothesis}, but we do not know whether they hold without any set theoretical hypothesis.
\end{abstract}

\baselineskip=13.5pt

\section{Introduction}

The motivation for this paper is a question which can be better formulated in terms of Social Choice Theory. Let us consider a committee  of $n$ members $c_1, \dots, c_n$ having to express  its preferences among $m$ candidates. Each member $c_k$  writes his own preferences among the $m$ candidates in a linearly ordered list $\ell_k$ of the candidates. The \emph{profile} of an ordered pair  $(x,y)$ of two different candidates $x$ and $y$ is the $0$-$1$ list  $(\ell_k(x,y))_{1\leq k\leq n}$, where $\ell_k(x,y)=1$ if $x$ is preferred to $y$, and $\ell_k(x,y)=0$ otherwise. As the profiles of the  $m(m-1)$ ordered pairs belong to  $\{0,1\}^n$,  if $m(m-1)> 2^n$, there are  two distinct ordered pairs with the same profile. The question is: does the converse hold? That is, if $m(m-1)\leq 2^n$, are there $n$ lists  $(\ell_k)_{1\leq k\leq n}$ yielding $m(m-1)$ distinct profiles? As we will see, the answer is positive for $m\leq 7$. For other integers we do not know.  

Tackling this question, we do not limit ourselves to finite sets. Considering a set $V$ of cardinality $m$, let $h_{\Lin}(m)$ be the least cardinal $n$ such that there is a family $\mathcal R$ of $n$ linear orders on $V$ such that any distinct ordered pairs $(x,y)$ and $(x',y')$ of distinct elements of $V$ yield distinct profiles. This parameter plays a role in the investigation of systems of hereditarily rigid relations started in Couceiro, Haddad, Pouzet and Sch\"olzel \cite{couceiro-haddad-pouzet-scholzel}.  An $h$-ary relation $\r$ on a  set $V$ is said to be \emph{hereditarily rigid} if the  unary partial functions on $V$ that preserve $\r$ are the subfunctions of the identity map or of constant maps. A family of relations ${\mathcal  R}$ is said to be \emph{hereditarily   rigid} if the unary partial functions on $V$ that preserve  every $\r \in {\mathcal  R}$ are the subfunctions of the identity map or of constant maps. As it turns out, a family of tournaments  ${\mathcal  R}$ is  hereditarily   rigid if and only if any two distinct ordered pairs $(x,y)$ and $(x',y')$ of distinct elements of $V$ yield distinct profiles of tournaments. We note that for $m<\aleph_0$  we may find such a family $\mathcal R$ made of $n$ tournaments if and only if $m(m-1)\leq 2^n$, that is $\log_2(m(m-1))\leq n$. We ask if the same inequality  holds  when  tournaments are replaced by linear orders, that is, wether  $h_{\Lin}(m)=\lceil \log_2(m(m-1))\rceil$.

We show that $h_{\Lin}(m)= \aleph_0$ if $\aleph_0\leq m\leq 2^{\aleph_0}$.  We show more generally that $h_{\Lin}(m)= \log_2 (m)= d(\Lin(V))$,  where $\log_2 (m)$ is the least cardinal $\mu$ such that $m\leq 2^\mu$ and  $d(\Lin(V))$ is the topological density of the set  $\Lin(V)$ of linear orders on $V$ (viewed as a subset of the  power set $\powerset (V\times V)$ equipped with the product topology). The last set of equalities   follows from   GCH (Generalized Continuum Hypothesis); we do not know if it holds without any set theoretical hypothesis. The finite case is more substantial, but apparently more difficult. In that direction, we verify  that  $h_{\Lin}(m)=\lceil \log_2(m(m-1))\rceil$ for $m\leq 7$ and prove that 
 $h_{\Lin}(m)\leq h_{\Lin}(m+1)\leq h_{\Lin}(m)+1$ for all $m<\aleph_0$.

 Notations in this paper are quite  elementary. The \emph{diagonal} of a set $X$ is the set $\Delta_X:= \{(x,x): x\in X\}$. We denote by  $\powerset (X)$ the collection of subsets of  $X$, by $X^m$ the set of $m$-tuples $(x_1, \dots,  x_m)$ of  $X$, by ${X \choose m}$ the subset of $\powerset (X)$ made of $m$-element subsets of $X$, and
by $[X]^{<\omega}$ the collection of finite subsets of $X$. The cardinality of  $X$ is denoted by $\vert X\vert$. If $\kappa$ denotes a cardinal, $2^{\kappa}$ is the cardinality of   the power set $\powerset (X)$ of a set $X$ of cardinality $\kappa$; we denote by $2^{< \kappa}$ the supremum of $2^{\mu}$ for $\mu<\kappa$. 
If  $\kappa$ is an infinite cardinal, we set $\log_2(\kappa)$ for the least cardinal $\mu$ such that $\kappa \leq 2^{\mu}$. If $\kappa$ is an integer, we use $\log_2(\kappa)$ in the ordinary sense, hence  the least integer $\mu$ such that $\kappa \leq 2^{\mu} $ is $\lceil \log_2\kappa  \rceil$.  We denote by $\aleph_0$ the first infinite cardinal.  We refer the reader to \cite{jech} for further background about axioms of set theory if needed. The proof of one  of our results (Theorem \ref{thm:hered-rigidity-bound})    
 relies on   the famous theorem of Sperner (see \cite{engel}). To state it, we recall that an \emph{antichain} of subsets of a set $X$ is a collection of subsets such that none is contained in another.

\begin{theorem}
  Let $n$ be a non-negative integer.  The largest size  of an  antichain family of subsets of  an $n$-element set  $X$ is $\binom{n}{\lfloor n/2 \rfloor}$. It is only realized by
 $\binom{X}{\lfloor n/2 \rfloor}$ and $\binom{X}{\lceil n/2 \rceil}$.
\end{theorem}

Let $2:= \{0, 1\}$ be ordered with $0<1$. The poset $2^{X}$ equipped with  the product order is isomorphic to the powerset $\powerset(X)$ ordered by inclusion. Also note that if $Y$ is any set, then  the posets  $(2\times 2)^ Y$,  $2^{Y}\times 2^{Y}$ and $2^{Y\times 2}$, all equipped with the product order, are isomorphic. If $\vert Y\vert =m$, $m\in \N$,  Sperner's theorem asserts  that the maximum sized antichain  in these posets, once   identified to  $0$-$1$-sequences, is made of sequences  containing roughly as many $0$ as $1$. This is the key for proving Theorem \ref{thm:hered-rigidity-bound}.

We present first the rigidity notions, then the case of tournaments and linear orders and we conclude with  density properties. 

\section{Hereditary rigidity}

In \cite {couceiro-haddad-pouzet-scholzel}, Couceiro {\it et al.}  studied a general notion of rigidity for relations and sets of relations w.r.t. partial operations. They show a noticeable difference between  rigidity w.r.t.  to unary operations and  rigidity  w.r.t. to operations of arity at least two. Here we consider the rigidity notion w.r.t.  unary operations, mostly when the relations are binary. 
Considering hereditarily rigid sets of binary relations,  we  give  an  exact upper bound on the size of their  domain (Theorem \ref{thm:hered-rigidity-bound}).

Let $V$ be a set.  
A \emph{partial function}  on $V$ is a map $f$ from a subset of $V$, its \emph{domain}, denoted by   $\dom(f)$,  to another, possibly different subset, its \emph{image}, denoted by   $\im(f)$.  A partial function $f$ is \emph{constant} if it does not have two distinct values. If $A$ is a subset  of $V$, the \emph{restriction}  of $f$ to $A$, denoted by $f_{\restriction A}$, is the map induced by $f$ on $A\cap \dom(f)$. A \emph{subfunction} of $f$ is any restriction of $f$ to a subset of its domain.

Let $h\ge 1$ be an integer, an  \emph{$h$-ary relation on} $V$ is a subset $\r$ of
$V^h$. Sometimes, we identify $\r$ with its characteristic function,  that is, we write  $\rho( v_1, \dots, v_h)=1$ if $( v_1, \dots, v_h)\in \r$ and $0$ otherwise. If $A$ is a subset of $V$, the restriction of $\r$ to $A$, denoted by $\r_{\restriction A}$,  is $\r\cap A^h$. If $\mathcal R$ is a set of relations on $V$, we set $\mathcal R_{\restriction A} := \{\r_{\restriction A}: \r\in \mathcal R\}$.

 We say that a partial function $f$ \emph{preserves} the $h$-ary relation $\r$, or $\r$ {\it is invariant under} $f$,
if for every $h$-tuple $(v_1, \dots, v_h)$ belonging to   $(\dom(f))^{h}\cap \r$, its image $(f(v_1), \dots, f(v_h))$ belongs to $\r$. We say that a partial function $f$ \emph{preserves} a  family of relations ${\mathcal  R}$ on $V$ if it preserves each $\r\in \mathcal R$. If $A:= \dom(f)$, we also  say that $f$ is a  \emph{homomorphism}  of $\mathcal R_{\restriction A}$ in $\mathcal R$.
 
A $h$-ary relation $\r$ on a  set $V$ is said to be {\em{rigid}} if  the identity map is the only unary function on $V$ that preserves $\r$. The relation $\r$  is  \emph{semirigid} if every unary function that preserves $\rho$ is the identity map or a constant map.  It is   \emph{hereditarily semirigid} if the  unary partial functions on $V$ that preserve $\r$ are the subfunctions of the identity map or of constant maps. A family of relations ${\mathcal  R}$ on $V$ is said to be \emph{hereditarily semirigid} if the unary partial functions on $V$ that preserve  every $\r \in {\mathcal  R}$ are the subfunctions of the identity map or of constant maps.  In order to agree with \cite {couceiro-haddad-pouzet-scholzel}, we delete the prefix "semi" in the sequel. 
Rigid binary relations are introduced in \cite{v-p-h}, semirigid relations in \cite{miyakawa2, lan-pos, miyakawa1,zadori}. 
\begin{proposition}\label{prop:hereditaryrigid}
For a family of relations ${\mathcal  R}$ on a set $V$ the following properties are equivalent:
\begin{enumerate}[(i)]
\item $\mathcal R$ is hereditarily  rigid; 
\item For every $2$-element subset $A$ of $V$, every homomorphism $f$ with domain $A$ of $\mathcal R_{\restriction A}$ in $\mathcal R$  is either constant or a subfunction of the identity map; 
\item For every $2$-element subset $A$ of $V$, every $1$-$1$ homomorphism $f$ with domain $f$ of $\mathcal R_{\restriction A}$ in $\mathcal R$  is  a subfunction of the identity map.  

\end{enumerate}
\end{proposition} 

\noindent \begin{proof}
Implications $(i)\Rightarrow (ii)$ and 
$(ii)\Rightarrow (iii)$ are immediate. We prove that implication 
$(iii)\Rightarrow (i)$ holds. Let $V' \subseteq V$ and $f$ be a  homomorphism with domain $V'$ of $\mathcal R_{\restriction V'} $ in $\mathcal R$. We need to prove that $f$ is either constant or a subfunction of the identity. We may suppose that $V'$ is not a singleton, otherwise, $f$ is constant. Suppose that  $f$ is $1$-$1$. Then $(iii)$ asserts that for every two-element subset $A$ of $V'$, $f_{\restriction A}$ is a subfunction of the identity map. It follows that $f$ is a subfunction of the identity and $(i)$ holds. If $f$ is not $1$-$1$ then  there is some $x\in V'$ such that $X:= f^{-1}(f (x))$ has a least two elements. We may suppose that $f(x)\not =x$. If there is  $y\in V' \setminus X$ then  $f_{\restriction \{x,y\}}$ is $1$-$1$ (indeed, $f(x)= f(y)$ amounts to $y\in X$ which is excluded)  and not the restriction of the identity since $f(x)\not = x$. Hence, $X= V'$ and $f$ is constant. 
\end{proof}

The above result has a particularly simpler form if the relations are binary and each one is either reflexive or irreflexive. To do this translation, we view such  binary relations on $V$  as maps from $V\times V\setminus \Delta_{V}$ to $2:= \{0,1\}$.

\begin{definition} Let $\mathcal R$ be a set of binary relations on $V$, each $\rho\in \mathcal R$  being either  reflexive  or irreflexive. Let  $(x, y) \in V\times V\setminus \Delta_{V}$. The \emph{profile} of $(x,y)$ with respect to $\mathcal {R}$ is $p _\mathcal R(x,y):= (\rho(x, y))_{\rho\in \mathcal R}$. The  \emph{profile} of $\mathcal R$ is the map $p_{\mathcal R}$ from $V\times V\setminus \Delta_{V}$ to $2^{\mathcal R}$, 
associating $p _\mathcal R(x,y)$ to each $(x,y)$. The \emph {double profile}   is the map  $\tilde p_{\mathcal R}$ associating the element $(\rho(x, y), \rho(y,x))_{\rho\in \mathcal R}$ of $(2\times 2)^{\mathcal R}$  to each ordered pair $(x,y)\in V\times V\setminus \Delta_{V}$. 
\end{definition}

Let $-$ be the involution defined on  $V\times V\setminus \Delta_{V}$ by $\overline {(x,y)}:= (y,x)$ for every $(x,y)\in V\times V\setminus \Delta_{V}$. Similarly,  let $-$ be the involution on  $2\times 2$  defined by  $\overline u:= (\beta, \alpha)$ for every $u:= (\alpha, \beta)\in 2\times 2$. If $\theta $ is any map from  a set $X$  to $2\times 2$,  let $\overline \theta$ be the composition of $\theta$ and $-$ , that is $\overline {\theta}(\r):= \overline {\theta(\r)}$ for every $\r\in X$. We say that a map  $\varphi:  V\times V\setminus \Delta_{V} \rightarrow (2\times 2)^X$ is \emph{self-dual} if $\varphi({\overline{(x,y)}})= \overline {\varphi (x,y)}$ for all $(x,y)\in \dom(\varphi)$.

%Order members of $2^{\mathcal R}$ componentwise and  members of  $2^{\mathcal R}\times 2^{\mathcal R}$ by the product ordering. 
\begin{lemma}\label{lem:cns}  Let  $V$ be a set. \begin{enumerate}
 \item A set $\mathcal R$ of binary relations on $V$, each one being either  reflexive  or irreflexive,  is hereditarily rigid if and only if $\tilde p_{\mathcal R}$, the double profile of $\mathcal R$, is $1$-$1$ and its range is an antichain of $(2\times 2)^{\mathcal R}$.  

\item Let $X$ be a set. If $\varphi$ is any $1$-$1$ self-dual map  from $V\times V\setminus \Delta_{V}$ to $(2\times 2)^{X}$  whose range is an antichain,  then  there is a map $\theta$ from $X$ onto a hereditary rigid  set  $\mathcal R$ of irreflexive binary relations on $V$ such that   the natural map $ \tilde{\theta}: (2\times 2) ^{\mathcal R} \rightarrow (2\times 2)^{X}$ defined by $\tilde{\theta}(\psi):= \psi\circ \theta$ satisfies  $\tilde \theta \circ  \tilde p_{\mathcal R}=\varphi.$
\end{enumerate}
\end{lemma} 
\begin{proof}
1) Observe that if $(x,y)$ and $(x',y')$ are in $V\times V\setminus \Delta_V$, the map transforming $x$ to $x'$ and $y$ to $y'$ is a  homomorphism of $\mathcal R_{\restriction \{x,y\}}$ to $\mathcal R$ if and only if  $(\rho(x, y), \rho(y,x))_{\rho\in \mathcal R}\leq (\rho(x', y'), \rho(y',x'))_{\rho\in \mathcal R}$. Hence,  the above condition on $\tilde p_{\mathcal R}$ amounts to $(iii)$ of Proposition \ref{prop:hereditaryrigid}. 

2) Let $p_1: 2\times 2 \rightarrow 2$ be the first projection,  let $\theta: X\rightarrow 2^{V\times V \setminus \Delta_V}$ defined by $\theta (u)(x,y):= p_1(\varphi (x,y)(u))$ for $u\in X$, $(x,y)\in V\times V \setminus \Delta_V$  and let $\mathcal R$ be the range of $\theta$.  
%2) Let  $\rho\in \mathcal R$ be a  binary relation; let  $\theta(\rho)$ and $\zeta(\rho)$ be binary relations such that $\varphi (x,y)(\rho)= (\theta(\rho)(x,y), \zeta(\rho)(x,y))$. We have $ \theta(\rho)(x,y)= \zeta(\rho)(y,x)$. The map $\theta$ yields the stated property. 
\end{proof}

\begin{theorem}\label{thm:hered-rigidity-bound}  
There is a hereditarily rigid set $\mathcal R$ of $\kappa$ binary relations, each one reflexive or irreflexive on a set $V$ of cardinality $\mu$ if and only if 
$\mu(\mu-1) \leq {\vert {2\kappa \choose {\kappa}} \vert }$ if $\kappa$ is finite  and 
$\mu \leq 2^{\kappa}$ otherwise.
\end{theorem}
\begin{proof}
According to $1)$ of Lemma \ref{lem:cns} and Sperner's Theorem the first inequality is satisfied. If $\kappa$ is infinite, we get  the second. For the converse, we define a $1$-$1$ self dual map $\varphi$ from $V\times V\setminus \Delta_{V}$ to $(2\times 2)^{X}$, where $\vert X\vert = \kappa$,  whose range is an antichain and  apply 2) of Lemma \ref{lem:cns}. For that,  let $\ell$ be a tournament on $V$. Due to Sperner's Theorem, we may choose a $1$-$1$ map $\varphi'$ from  $\ell$ to the middle level of $(2\times 2)^{X}$. Next, select an involution  $\sigma$ on this middle level with no fixed point (e.g. associate to each $0$-$1$-sequence the sequence obtained by exchanging the $0$ and $1$). Then, set $\varphi(x,y):= \varphi'(x,y)$ for $(x,y)\in \ell$  and $\varphi(x,y):= \sigma(\varphi'(y,x))$ otherwise.  This map is self-dual. \end{proof}
We examine the case of tournaments and linear orders in the next two sections.
\section{Separation and Hereditary rigidity of tournaments} 

Let $V$ be a set.  A \emph{tournament} on $V$ is an irreflexive  binary relation $\tau$ on $V$ such that for every ordered pair $(x,y)$ either $(x,y)\in \tau$ or $(y,x)\in \tau$, but not both. 

Let $\Tour(V)$ be the set of tournaments on $V$.  We say that a tournament  $\tau$  \emph{separates} two distinct ordered pairs $(x,y), (x',y')\in V\times V\setminus \Delta_{V}$ if $\tau (x,y)\not = \tau(x',y')$. 

Despite that fact that linear orders are reflexive, and tournaments are not, we may view linear orders as tournaments and apply to them what follows. 
 
 \begin{lemma}\label{lem:linorderseparate}Let $V$ be a set; then two distinct pairs $(x,y), (x',y')\in V\times V\setminus \Delta_{V}$ are always separated by some linear order.
 \end{lemma} 
 \begin{proof}Indeed, if $(x,y)= (y',x')$ any linear order containing $(x,y)$ will do. If not then the reflexive transitive closure of $\{(x,y), (y',x')\}$ is an order. Any linear extension of that order will do.\end{proof}

\begin{lemma} \label{lem:tour}Let $\mathcal R$ be a family of tournaments on a set $V$. The following properties are equivalent:
\begin{enumerate}[{(i)}]\item For all distinct ordered pairs $(x,y), (x',y')\in V\times V\setminus \Delta_{V}$ there is always some member of $\rho \in \mathcal R$ that separates them;
\item The family  $\mathcal R$ is hereditarily rigid. 
\end{enumerate}
\end{lemma} 

\begin{proof}
$(i)\Rightarrow (ii)$. Let $U$ be a two-element subset of $V$ and $f$ be a partial  homomorphism of $\mathcal R$ defined on $U$. Supposing $f$ non constant, we prove that $f$ is the identity. Let $x, y$ be the two elements of $U$, let $x':= f(x)$, $y':= f(y)$.  If the ordered pairs $(x,y)$ and $(x',y')$ are distinct, they are separated by some $\ell\in \mathcal R$, i.e., verifying $\ell(x,y)\not = \ell(x',y')$. Since $f$ is an endomorphism, if $\ell(x,y)= 1$ then $\ell (x',y')=1$. Thus $\ell(x,y)=0$. Since $\ell$ is a tournament, $\ell(y,x)=1$ and since $f$ is an endomorphism, $\ell(y',x')=1$, but then $\ell(x,y)=\ell(x',y')=0$, contradicting the fact that $\ell$ separates $(x,y)$ and $(x',y')$.\\ 
\noindent $(ii)\Rightarrow (i)$. Let $(x,y), (x',y')$ be two distinct irreflexive ordered pairs. The $1$-$1$ map $f$ defined on $U:= \{x, y\}$ such that $f(x):= x', f(y):= y'$ is not the identity, hence it cannot be a  homomorphism; so there is some $\ell\in \mathcal R$ which is not preserved by $f$, meaning that there is some $(u,v)\in \ell$ such that $(f(u),f(v))\not \in \ell$. If $(u,v)= (x,y)$ then $1= \ell(x,y)\not = \ell(x',y')=0$ while if $(u,v)= (y,x)$, then since $\ell$ is a tournament, $0= \ell(x,y)\not = \ell(x',y')=1$, proving that $\ell$ separates these  two ordered pairs. 
\end{proof}

\begin{definitions}\label{def:h}
Let $V$ be a set,  $\kappa$ be its cardinality (possibly infinite) and  $\mathcal R$ be a set of tournaments on $V$ satisfying one of the equivalent conditions of Lemma \ref{lem:tour}.  We define $h_{\mathcal R}(\kappa)$ as the least cardinal $\mu$ such that there is some subset $X$ of $\mathcal R$ of cardinality $\mu$ such that all distinct ordered pairs $(x,y), (x',y')\in V\times V\setminus \Delta_{V}$ are always separated by some member of $X$.  Let $X\subseteq \mathcal R$. 
Fix $\ell\in  X$. The \emph{profile of} $X$ \emph{with respect to} $\ell$ is the family $p_{\ell} (X):=\{p_X(x,y): (x,y)\in \ell\}$. 
 This profile is \emph{minimal} if $p _{X}(x,y)\not = p_{X}(x',y')$ for any two distinct ordered pairs $(x,y), (x',y')\in \ell$.
\end{definitions}

\begin{lemma} Let $\ell\in X$  and $p_{\ell} (X)$ be the  profile of $X$ with respect to $\ell$. Then
  $p_{\ell} (X)$ is minimal if and only if  all distinct ordered pairs $(x,y), (x',y')\in V\times V\setminus \Delta_{V}$ are always separated by some member of $X$.
\end{lemma}

\begin{proof} Suppose that $p_{\ell} (X)$ is minimal. Let $(x, y), (x',y') \in V\times V \setminus \Delta_{V}$ be two distinct ordered pairs. If  $\ell$ separates these pairs, we are done. Otherwise $\ell(x,y)= \ell(x',y')$. If the common value is $1$, then  since $p_{X}(x,y)\not = p _X(x',y')$ there is some $\ell'\in X\setminus \{\ell\}$ such that $\ell'(x,y) \not =\ell'(x',y')$. If the common value is $0$, then $(y,x), (y',x')\in \ell$ and the previous reasoning  yields the same conclusion.  Suppose that the separation property holds. Then two distinct ordered pairs $(x,y), (x',y')\in \ell$ are separated by some member of $X$, thus $p _X(x,y)\not = p _X(x',y')$. 
\end{proof} 

An immediate corollary is the following. 
\begin{corollary} If the profile of $X$ with respect to $\ell\in X$ is minimal, then its profile with respect to any other  $\ell'\in X$ is minimal too. 
\end{corollary}
Another straightforward consequence is the following result. 

\begin{proposition}\label{prop:minimality}
Let $V$ be a set of cardinality $\kappa$,  $\mathcal R$ be a set of tournaments on $V$ satisfying one of the equivalent conditions of Lemma \ref{lem:tour}. Then $h_{\mathcal R} (\kappa)$ is the minimum of the cardinality of a subset $X$ of $\mathcal R$ such that its profile with respect to some tournament $\ell\in X$ is minimal.  \end{proposition}

\begin{lemma}\label{lem:minoration}Under the conditions of Definition \ref{def:h} the following inequality holds:
$h_{\mathcal R}(\kappa)\geq \log_2(\kappa \cdot (\kappa-1)).$
\end{lemma}

\begin{proof} Let $V$ be a set, $X$ be a subset of $\Tour(V)$. Suppose that $p_{\ell} (X)$ is minimal. Associate to  each $(x,y)\in \ell$ the profile of $X\setminus\{\ell\}$, that is $p _{X\setminus \{\ell\}}(x,y)$. This defines a map from $\ell$ into   $2^{X\setminus \{\ell\}}$. This map being $1$-$1$,  we have $\vert \ell \vert \leq 2^{\vert X\vert -1}$, that is $\frac{\kappa.(\kappa-1)}{2} \leq  2^{\vert X\vert-1}$. This amounts to $\kappa\cdot (\kappa-1)\leq 2^{\vert X\vert}$, that is $\log_2(\kappa\cdot(\kappa-1))\leq \vert X\vert$.
\end{proof}

We show in Theorem \ref{thm:exactvalue} below that the equality holds
for $\mathcal R= \Tour(V)$ but  for $\mathcal R:= \Lin (V)$ the exact value of $h_{\mathcal R}(n)$ for $n\in \N$ eludes us.

\begin{theorem}\label{thm:exactvalue}
$h_{\Tour} (\kappa)= \log_2(\kappa)$ if $\kappa$ is an infinite cardinal and $h_{\Tour} (\kappa)= \lceil \log_2(\kappa\cdot(\kappa-1)) \rceil$ if $\kappa$ is a non negative integer.
\end{theorem} 
\begin{proof} From the lemma above we have $h_{\Tour} (\kappa)\geq  \log_2(\kappa \cdot (\kappa-1))$. For the reverse inequality, let $Z$ be  a subset  of cardinality $\lceil \log_2(\kappa\cdot (\kappa-1)) \rceil$ of $\Tour(V)$. Fix $\ell\in Z$. Choose a $1$-$1$ map $\varphi$ from $\ell$ into $2^{Z\setminus \{\ell\}}$. For $k\in Z\setminus \{\ell\}$ set  $\ell_k:  = \{(x,y)\in \ell: \varphi(x,y)(k)=1\}\cup\{(x,y): (y,x)\in \ell \; \text{and}\; \varphi(y,x)(k)=0\}$.  It is straightforward to check that  $X:= \{\ell_k: k\in Z\setminus \{\ell\}\}\cup \{\ell\}$ is a separating family of tournaments. Inequality  $h_{\Tour} (\kappa)\leq  \log_2(\kappa\cdot (\kappa-1))$ follows.
\end{proof}

\begin{remark}
We could choose $Z\subseteq Lin(V)$ in the proof above. However, there is a priori no much relationship  between $Z$ and $X$. 

\end{remark}
\section{The case of linear orders}\label{linear order}

On $\{1, \dots, m\}$ the $m$ cyclic permutations of the natural order form  a separating family, hence  $h_{\Lin} (m)\leq m$ for every integer $m$. Due to the minoration of  $h_{\Lin}(m)$ by $\lceil \log_2(m\cdot (m-1))\rceil$ we get the equality for $3\leq m\leq 5$. Here is an example of  a separating  family of $5$ linear orders on a $6$-element set proving that
$h_{\Lin} (6)=5$. We give these orders  by the five following strings:
$123456; 136542;  216543; 432165;  532146$.  

With the proposition below we get $h_{\Lin} (7)=6$. We do  not know if $h_{\Lin} (8)=6$. 
   
\begin{proposition} Let $m\in \N$. Then $h_{\Lin}(m) \leq h_{\Lin}(m+1)\leq h_{\Lin}(m)+1$. 
\end{proposition}
\begin{proof}The first inequality is trivial. For the  second,  let $V:= \{1, \dots, m\}$. Let $\mathcal R:= (\leq_k)_{1\leq k\leq n}$, with $n:=h_{\Lin}(m)$,  be a separating family of linear orders $\leq_k$ on $V$.  Our aim is to extend these linear orders on $V\cup \{m+1\}$ and,  with an extra linear order,   obtain a separating family.  We suppose  that   $\leq_1$ is the natural order on $V$ and we add $m+1$ just after $m$ in $\leq_1$.   For each $k$, $2\leq k\leq n$, we insert $m+1$ just before or just after $m$ in $\leq_k$. These choices are decided by a $0-1$-sequence $s(m, m+1)$   of length $n-1$ that we are going to define. For $1\leq i<j\leq m$,  let $s(i,j):= (s_k(i,j))_{2\leq k\leq n}$, where $s_k(i, j)=1$ if $i<_kj$ and $0$ otherwise. Since $\mathcal R$ is separating, $m(m-1) \leq 2^n$, hence $2^{n-1}-\frac{m(m+1)}{2}  \geq  p:= 2^{\frac{n}{2} -1}\geq 1$. Hence,  there are at least $p$ $0$-$1$-sequences  of length $n-1$,  which are distinct of all the $\frac{m(m-1)}{2}$ sequences $s(i,j)$. Let $s(m,m+1)$ be such a sequence. Extend the orders as said, and add a new  linear order, say  $\leq_{n+1}$, for which $m+1$ is the least element and $m$ the last one. We show that all  new profiles, still denoted $s(i,j)$,  are distinct, thus proving that the new family of orders is separating. The claimed inequality follows from the following observations. 

1) $s(i, m) \not = s(i,m+1)$ for   $i=2, m-1$   since  
$s_{n+1}(i,m+1)=0\not = 1= s_{n+1} (i, m)$.  

2) $s(m, m+1)\not =s(i, j)$ for $i<j\leq m$. This is just the choice of  $s(m, m+1)$. 

3) $s(i, m+1) \not = s(j, m+1)$ for $i<j<m+1$. Indeed, since for every $k$, $2\leq k\leq n$,  $m+1$ is immediately before or after $m$ in $\leq_k$ we have $s_k(h, m)  = s_k(h, m+1)$ for all $h<m$. Hence $s(h, m)= s(h, m+1)$,  and in particular  $s(i, m)= s(i, m+1)$ and $s(j, m)= s(j, m+1)$. Since $\mathcal R$  is separating,  we have $s(i, m)\not =s(j,m)$, hence $s(i, m+1) \not = s(j, m+1)$. \end{proof}

\section{Density}\label{section:density}

Let $T$ be a topological space and $D$ be a subset of $T$. An element $x$ of $T$ is \emph{adherent} to  $D$ if every open set containing $x$  meets $D$. The \emph{topological closure} (or the \emph{adherence}) of $D$  is the set of  elements of $T$ adherent to $D$, denoted by $\overline D$; this  is the least  closed set containing $D$. The subset $D$ is   \emph{dense} in $T$ if $\overline D=T$. The \emph{density character} of $T$,  denoted by $d(T)$,  is the minimum cardinality of a dense subset. Let $V$ be a set; if $\ell$ is a linear order on $V$, we may equip $V$ with the \emph{interval topology}  whose open sets are generated by the \emph{open intervals} of this order, i.e., sets of the form $]a, b[$ with $a<_{\ell} b$. We will denote by $d(V, \ell)$ the density of this space. We may note that the density of any subset with the interval topology is at most the density of $(V, \ell)$.  
We  may equip the power set $\powerset (V)$  with the product topology, a basis of open sets being  made of sets of the form $O(F, G):= \{X\in \powerset (V): F\subseteq X\subseteq V  \setminus G\}$ where $F$, $G$ are finite subsets of $V$. This is the well known Cantor space. A basic result in topology due to Hausdorff  (1936) (see \cite{vandouwen} in the Handbook on Boolean Algebras, vol. 2 p. 465)  asserts that $d(\powerset (V))= \log_2(\vert V\vert)\;  \text{provided that}\;  V\;  \text{is infinite}$.   

Let $m$ be an integer and $\mathcal R$ be  a set of $m$-ary  relations on a set $V$;  each $\r\in \mathcal R$ can be  viewed as a map from $V^m$ to $\{0,1\}$ as well as a subset of $V^m$. Viewing $\mathcal R$  as a subset of $2^{V^m}$, we  may equip $\mathcal R$ of the topology induced by the product topology on $2^{V^m}$.  Let $D$ be a subset of $\mathcal R$ and $k$ be an integer. A relation $\r\in \mathcal R$ is \emph{$k$-adherent} to  $D$ if on  every $k$-element subset $F$ of $V^m$, $\r$ coincides with some $\r'\in D$. The \emph{$k$-adherence of $D$ in $\mathcal R$} is the set of $\r\in \mathcal R$ that are $k$-adherent to $D$. If this set is $\mathcal R$, $D$ is said \emph{$k$-dense}. The relation $\r$ is  adherent to $D$ in the topological sense if it is $k$-adherent for every integer $k$. The adherence of $D$ is often called the \emph{local closure} of $D$. 
The topological density of several sets of relations can be computed.  For example, the following sets of relations on a set of cardinality $\kappa$ have density $\log_2(\kappa)$:
the collection of $n$-ary relations, of binary relations, of directed graphs without loops, of undirected graphs, of tournaments. Indeed, all of these sets are homeomorphic to the powerset of some set of cardinality $\kappa$. We just illustrate  this fact with the collection of tournaments.  Fix a tournament $\tau: =(V, E) $ on a set $V$ of size $\kappa$. To each map $f: E\rightarrow 2$ associate the tournament $\tau_f$ whose arc set is $E_f:=f^{-1} (1)\cup \{ (u, v)\in V^2\setminus {\Delta_V}: f(v,u)=0\}$. This defines a  homeomorphism from $2^E$ onto the set of tournaments. \\

If  $X$ is  a set of binary relations $\r$ on $V$, we set $X^{-1}:= \{\r^{-1}: \r \in X\}$. 

\begin{lemma}\label{lem:separation}
Let $\mathcal R $ be a collection  of tournaments on $V$ such that two distinct ordered pairs $(x,y), (x',y')\in V\times V\setminus \Delta_{V}$ are always separated by some member of $\mathcal R$.  The following properties are equivalent for a non-empty subset $X$ of $\mathcal R$:  
\begin{enumerate}[{(i)}]\item All distinct ordered pairs $(x,y), (x',y')\in V\times V\setminus \Delta_{V}$ can be   always separated by some member of $X$;
\item $X\cup X^{-1}$ is $2$-dense in $\mathcal R$. 
\end{enumerate}
\end{lemma} 
\begin{proof}
$(i)\Rightarrow (ii)$ Let $\rho \in \mathcal R$. Let $U$ be a two element subset of $V\times V\setminus \Delta_{V}$. We need to prove  that some $\tau\in X\cup X ^{-1}$ coincides with $\rho$ on $U$.  
 Let $(x,y)$ and $(x',y')$ be two distinct ordered pairs such that $U= \{(x,y),  (x',y')\}$. Let $\alpha:= \rho(x,y)$ and $\beta:=\rho(x',y')$. If $(\alpha, \beta)=(0,0)$, let $\tau$ be a tournament separating   $(x,y)$  and $(y',x')$. Then $\tau$ or $\tau^{-1}$ coincides with $\rho$ on $U$. If $(\alpha, \beta)=(0, 1)$, let $\tau$ separating   $(x,y)$ and $(x',y')$. Then $\tau$ or $\tau^{-1}$ coincides with $\rho$. The two other cases are similar. 
 
$(ii)\Rightarrow (i)$ Let $(x,y)$ and $(x',y')$ be two distinct ordered pairs. These pairs are separated by some member $\rho$ of $A$. Since $X\cup X^{-1}$ is $2$-dense in $A$ there is some  $\tau \in X\cup X^{-1}$ which coincides with $\rho$ on $U: =\{(x,y), (x',y')\}$. Hence $\tau$ separates $(x,y)$ and $(x',y')$.  \end{proof}

Let $\Lin(V)$ be the set of linear orders on an infinite  set $V$ of cardinality $\kappa$. We may identify each linear order with a subset of $V\times V\setminus \Delta_V$. By definition a subset $X$ of $\Lin (V)$ is dense if every nonempty open set of $\Lin (V)$ meets $X$. This amounts to the fact that for every $\rho \in \Lin(V)$,  every  finite subsets $F$, $G$ of $V\times V\setminus \Delta_{V}$), every set  $O(F, G):= \{\tau\in \powerset (V\times V\setminus \Delta_{V} ): F\subseteq \tau\subseteq V\times V\setminus G\}$ containing $\rho$ meets $X$. 

According to Lemma \ref{lem:linorderseparate}, we may apply Lemma \ref{lem:separation}.  This yields: 
\begin{lemma}\label{lem:density} If $X$ is dense in $\Lin (V)$, then $X$ separates all distinct ordered pairs. 
Hence:
$h_{\Lin(V)} (\vert V\vert ) \leq d(\Lin (V)$.
\end{lemma}

An alternative condition to density is this:
\begin{lemma}\label{prop:density} $X$ is dense in $\Lin (V)$ if and only if for every non-negative integer $m$  and every $m$-tuple $(a_1, \dots, a_m)$  of distinct elements of $V$ there is some $\ell\in X$ such that $a_1<_{\ell} \dots <_{\ell}a_m$. 
\end{lemma}
\begin{proof}Suppose that $X$ is dense. Let $m\in \N$ and  a $m$-tuple $(a_1, \dots a_m)\in V^m$ with distinct entries. Let $\rho\in \Lin (V)$ such that  $a_1<_{\rho}\dots <_{\rho}a_m$. Set $F:= \{ (a_i, a_{i+1}): 1\leq i<m\}$ and $G:= \emptyset$. Then $\r\in  O(F, G)$ thus there is some $\ell \in X\cap O(F, G)$ that is  $a_1<_{\ell}\dots  <_{\ell}a_m$. Conversely, let $\r\in O(F,G)\cap \Lin (V)$. We prove that $O(F,G)$ meets $X$. Let $A$ be a finite subset of $V$  such that $\bigcup \{\{x,y\} : (x,y) \in F\cup G\}\}\subseteq A$. Let $(a_1, \dots a_m)$ be an enumeration of elements of $A$ so that $a_1<_{\r}\dots<_{\r}a_m$. Then there is $\ell \in X$ such that the ordering on $A$ coincides with this of $\r$. It follows that $\ell \in O(F,G)\cap X$.  \end{proof}

It is easy to show that  $d(\Lin(V))=\log_2(\vert V\vert)$  if  $\aleph_0\leq \vert V\vert \leq 2^{\aleph_0}$ (e.g., see the proof of Theorem \ref{thm:cardinalinvariant} below). Provided that GCH  holds, the  equality $d(\Lin(V))=\log_2(\vert V\vert)$ holds for every infinite set; but we do not know if this holds  without any set theoretical hypothesis.

To prove  that this equality holds, we introduce the following related parameters. 
For a cardinal $\kappa$, let $\delta(\kappa)$ be the least cardinal $\mu$ such that there is a linear order $\ell$ on a set of cardinality $\kappa$ admitting a dense subset  of cardinality $\mu$ and for a cardinal $\mu$, let $\ded (\mu)$ be the supremum of cardinals $\kappa$ such  that  there is a  chain of cardinality $\kappa$ and density at most $\mu$. Recall that if GCH holds then $2^{<\nu}= \nu$ for every infinite cardinal $\nu$. Next, under this cardinality condition, $\delta (\kappa)= \log_2(\kappa)$ and $\ded(\mu)= 2^\mu$ for every infinite $\kappa$ and  $\mu$. (hint: let $W$ be a well ordered chain of  cardinality $\mu$ and $2^{W}$ be the power set lexicographically ordered. The density of this chain, as any subchain,  is at most $2^{<\mu}$. By definition, if  $\mu:= \log_2( \kappa)$, $2^{< \mu} \leq \kappa\leq 2^{\mu}$, thus $\delta(\kappa)\leq 2^{<\mu}$. If  $2^{<\mu}= \mu$ then  $\delta (\kappa)\leq \mu$ and hence $\delta(\kappa)= \log_2(\kappa)$. The proof of the second equality is similar, and shorter). Mitchell \cite{mitchell} showed that it is consistent with ZFC that for uncountable regular $\mu$, $\ded(\mu)<2^\mu$ \cite{mitchell} and in fact that  no chain  of cardinality $2^{\kappa}$ has a  dense subchain of cardinality $\kappa$. In particular,  $\delta( \kappa)\not = \log_2(\kappa)$.   
Now we prove:
\begin{theorem}\label{thm:cardinalinvariant} For   every infinite cardinal $\kappa$ and set $V$  of cardinality $\kappa$, we have \label{dedekind}   $\log_2(\kappa) \leq h_{\Lin} (\vert V\vert ) \leq d(\Lin (V))\leq  \delta(\kappa)$. 
\end{theorem}

\begin{proof} 
The first inequality is Lemma \ref{lem:minoration} and the second is Lemma \ref{lem:density}.  We prove that the third inequality holds. For that, let $\ell:= (V, \leq)$ be a linear order having a dense set $D$ of cardinality $\delta(\kappa)$. For each finite subset $F$ of 
$D$, write its  elements  in an increasing order,  $F:= d_1<\dots <d_{m-1}$ with $m=\vert F\vert +1$  and   decompose   $\ell$ into the $m$ intervals $\ell_1:= ]-\infty, d_1[, \dots, \ell_{i+1}:= [d_i, d_{i+1}[, \dots, \ell_{m}:= [d_{m-1}, +\infty[ $. Let $\mathfrak S_{m}$ be the set of permutations of $\{1, \dots, m\}$. Each permutation $\sigma \in \mathfrak S_{m}$ induces a permutation of the intervals $\ell_1, \dots, \ell_{m}$. The lexicographical sum 
$\ell_{\sigma}: = \sum_{1\leq i\leq m} \ell_{\sigma(i)}$ yields a linear order $\leq_{\sigma}$ on $V$.

\noindent{\bf Claim.} The set $\mathcal C:= \{\leq_{\sigma}: F\in D^{<\omega} \text{and}\; \sigma\in \mathfrak S_{\vert F\vert +1}\}$ is dense in $\Lin(V)$. 
 
Since the number of pairs $(F, \sigma)$ where $F\in  D^{<\omega}$ and $\sigma \in \mathfrak S_{\vert F\vert +1}$ is $\vert D\vert$  this claim suffices to prove the inequality. 

\noindent {\bf Proof of the Claim. }
According to Proposition \ref{prop:density} this amounts  prove that for every integer $m$  and every $m$-tuple $(a_1, \dots a_{m})$  of distinct elements of $V$ there is some $\leq_{\sigma} \in \mathcal C$ such that $a_1<_{\sigma}\dots <_{\sigma}a_{m}$. 
Let $\sigma\in \mathfrak S_{m}$  such that $a_{\sigma^{-1}(1)}<\dots <a_{\sigma^{-1}(m)}$ in the chain $\ell$. Due to the density of $D$, there are  $d_{i}\in  [a_{\sigma^{-1}(i)}, a_{\sigma^{-1}(i+1)}[\cap D$ for $1\leq i\leq m-1$. Let $\ell_{1}:= ]-\infty, d_1], \dots, \ell_{i+1}:= ]d_i, d_{i+1}], \dots, \ell_{m}:= ]d_{m-1}, +\infty[$. Permute these intervals according to $\sigma$. The resulting   chain $\ell_{\sigma}$ reorders the $m$-tuple $a_{\sigma^{-1}(1)}<\dots <a_{\sigma^{-1}(m)}$ as 
$a_1<_{\sigma}\dots  <_{\sigma}a_{m}$. This proves our claim.\end{proof} 

\begin{corollary} 

$\log_2(\kappa) = h_{\Lin} (\vert V\vert )=d(\Lin (V))= \delta(\kappa) \;  \text{iff}\;  \log_2(\kappa)=\delta(\kappa)$. 

\end{corollary}

\section*{Acknowledgments} {This research was completed  while the third  author  of this paper stayed at the University of Tsukuba from May 8 to July 6, 2019. Support provided by  the  JSP is gratefully acknowledged. 

The authors are grateful to the referees of this paper for their careful examination, corrections and suggestions.}

 \end{document}